\documentclass[11pt]{llncs}
\usepackage[utf8]{inputenc}
\usepackage{algorithm}
\usepackage{algorithmicx}
\usepackage[noend]{algpseudocode}
\usepackage{algpascal}
\usepackage[dvipdfm]{graphicx}
\usepackage{multirow}
\usepackage{cite}
\usepackage{enumerate}

\usepackage{amsmath}
\usepackage{amssymb}
\usepackage{color}

\newcommand{\abs}[1]{\left\vert#1\right\vert}
\newcommand{\set}[1]{\left\{#1\right\}}
\newcommand{\comb}[2]{\left(\begin{array}{c}#1 \\ #2 \end{array} \right)}
\newcommand {\defined} {\stackrel{def} {=}}

\newcommand\restr[2]{{
  \left.\kern-\nulldelimiterspace 
  #1 
  \vphantom{\big|} 
  \right|_{#2} 
  }}

\newcommand{\commentfig}[1]{#1}
\newcommand{\runningtitle}[1]{\vspace{0.5ex}\noindent{{\textbf{\boldmath #1:~}}}}

\newcommand{\npc}{\textsc{NP}\textrm{-Complete}}
\newcommand{\nph}{\textsc{NP}\textrm{-Hard}}
\newcommand{\apxh}{\textsc{APX}\textrm{-Hard}}
\newcommand{\apx}{\textsc{APX}}

\newcommand{\ptas}{\textsc{PTAS}}

\newcommand{\nohyp}{\textsc{P}=\textsc{NP}}

\renewcommand{\root}[1]{root(#1)}
\newcommand{\inn}[2]{in_{#1}(#2)}
\newcommand{\parent}[2]{parent_{#1}(#2)}
\newcommand{\children}[2]{chld_{#1}(#2)}

\newcommand{\ca}{{\mathcal A}}
\newcommand{\cc}{{\mathcal C}}
\newcommand{\bigoh}{{\mathcal O}}
\newcommand{\pp}{{\mathcal P}}
\newcommand{\ff}{{\mathcal F}}
\newcommand{\cs}{{\mathcal S}}
\newcommand{\TT}{{\mathcal T}}
\newcommand{\traversals}{tr}
\newcommand{\tc}{{tc}}
\newcommand{\rc}{{rc}}
\newcommand{\co}{{cc}}

\newcommand{\lightksubgraph}{\textsc{k-Lightest Subgraph}}

\newcommand{\minrc}{\textsc{MinRCEC}}
\newcommand{\mincc}{\textsc{MinCCEC}}
\newcommand{\minrcpt}{\textsc{MinRCPTEC}}
\newcommand{\mincca}{\textsc{MinCCAEC}}
\newcommand{\minccainst}{\textsc{$(G,r,X, tc)$}}
\newcommand{\minsc}{\textsc{MinSC}}
\newcommand{\minthreesctwo}{\textsc{Min3SC2}}

\newcommand {\dsum} {\displaystyle \sum}
\newcommand {\instrc} {(\mathcal{P},X,\tc)}
\newcommand {\instrcpt} {(G,r,X,\tc)}

\begin{document}

\title{Edge Coloring with Minimum Reload/Changeover Costs}

\author{Didem Gözüpek \inst{1} \and  Mordechai Shalom \inst{2,3}}

\institute{Department of Computer Engineering, Gebze Technical University, Kocaeli, Turkey. \email{didem.gozupek@gtu.edu.tr}
\thanks{This work is supported by the Scientific and Technological Research Council of Turkey (TUBITAK) under grant no. 113E567.}
\and
TelHai Academic College, Upper Galilee, 12210, Israel. \email{cmshalom@telhai.ac.il}
\and
Department of Industrial Engineering, Boğaziçi University, Istanbul, Turkey.
\thanks{Supported in part by the TUBITAK 2221 Programme.}
}

\pagestyle{headings}
\maketitle

\begin{abstract}
In an edge-colored graph, a traversal cost occurs at a vertex along a path when consecutive edges with different colors are traversed. The value of the traversal cost depends only on the colors of the traversed edges. This concept leads to two global cost measures, namely the \emph{reload cost} and the \emph{changeover cost}, that have been studied in the literature and have various applications in telecommunications, transportation networks, and energy distribution networks. Previous work focused on problems with an edge-colored graph being part of the input. In this paper, we formulate and focus on two pairs of problems that aim to find an edge coloring of a graph so as to minimize the reload and changeover costs. The first pair of problems aims to find a proper edge coloring so that the reload/changeover cost of a set of paths is minimized. The second pair of problems aim to find a proper edge coloring and a spanning tree so that the reload/changeover cost is minimized. We present several hardness results as well as polynomial-time solvable special cases.

\end{abstract}


\section{Introduction}\label{sec:Intro}
\subsection{Background}
The cost incurred while traversing a vertex via two consecutive edges of different colors is called \emph{traversal cost} and it depends only on the colors of the traversed edges. This cost leads to two different global cost measures that appeared in the literature under the names of \emph{reload cost} and \emph{changeover cost}.

The reload cost concept is defined in \cite{wirth2001reload} and it received attention only recently, although it has numerous applications. For instance, a color may represent the mode of transportation in an intermodal cargo transportation network. The traversal cost corresponds to the cost of transferring cargo from one carrier to another. Another application is in energy distribution networks, where the energy transfer from one carrier to another one, such as the conversion of natural gas from liquid to gas state, results in loss of energy. In telecommunications, traversal costs arise in numerous settings. For instance, routing in a heterogeneous network requires switching among different technologies such as cables, fibers, and satellite links. This switching cost can be modeled by traversal costs. Even within the same technology, switching between different providers, for instance switching between different commercial satellite providers in satellite networks, leads to a switching cost. All applications hitherto mentioned can be modeled using traversal costs where an edge-colored graph is given as input, and this is the focus of the works in the literature, e.g. \cite{wirth2001reload, galbiati2008complexity, GGM14mincyclecover, gourves2010dmtcs, AGM11minreloadpathstoursflows, gamvros2012reload, GGM11minchangeover, GSVZ14}. However, problems of finding a proper edge coloring so as to minimize the reload (or changeover) cost have important applications as well. For instance, recently, cognitive radio networks (CRN) have gained increasing attention in the communication networks research community. Unlike other wireless technologies, CRNs are envisioned to operate in a wide range of frequencies. Therefore, switching from one frequency band to another frequency band in a CRN has a significant cost in terms of delay and power consumption \cite{gozupek2013spectrum}. An optimal allocation of frequencies to the wireless links in CRNs so that the switching cost is minimized, corresponds to a proper edge coloring minimizing the traversal cost. In this work we focus on this type of problems.

The reload cost refers to the sum of the traversal costs of a set of paths in a graph, whereas the changeover cost does not depend on the amount of commodity (or number of paths) traversing a vertex. The works \cite{wirth2001reload,galbiati2008complexity} consider the problem of finding a spanning tree having minimum diameter with respect to reload cost. The paper \cite{GGM14mincyclecover} considers the \emph{minimum reload cost cycle cover} problem, which is to find a set of vertex disjoint cycles spanning all vertices, with minimum reload cost. In \cite{gourves2010dmtcs}, the authors study the problem of finding a path, trail, or walk of minimum reload cost between two given vertices. They consider (a)symmetric reload costs and reload costs with(out) triangle inequality.
The paper \cite{gamvros2012reload} studies the problem of finding a spanning tree that minimizes the sum of reload costs over the paths
between \emph{all} pairs of vertices. The paper \cite{AGM11minreloadpathstoursflows} presents various path, tour, and flow problems related to reload costs.
One of the problems studied in that work is the \emph{minimum reload cost path-tree}, which
is to find a spanning tree that minimizes the reload cost from a given root vertex to all other vertices.
The paper \cite{GGM11minchangeover} studies a closely related, yet different problem, called
\emph{minimum changeover cost arborescence} in which the goal is to find a spanning tree that minimizes the changeover cost from a given root vertex to all other vertices.
The work in \cite{GSVZ14} considers the same problem and derives several inapproximability results, as well as polynomial-time algorithms for some special cases. In \cite{GSSZ14-ChangeoverTreewidth} these special cases are extended to show that the problem is polynomial-time solvable in bounded treewidth graphs.

\subsection{Our Contribution and Research Directions}
In this paper, we consider a different problem and focus on proper edge coloring of a given graph such that the reload/changeover cost is minimized. To the best of our knowledge, this paper is the first study of this family of problems. Specifically, we formulate two pairs of problems. In the first pair of problems, given a set of paths comprising a graph, the goal is to find a proper edge coloring of the graph so that the reload (resp. changeover) cost is minimized. In the second pair of problems, given a graph and a root vertex, the goal is to find proper edge coloring and a spanning tree of the graph so that the reload (resp. changeover) cost from the root vertex to all other vertices is minimized. We present several hardness results as well as polynomial-time solvable special cases of these problems.

Specifically, we show that the first pair of problems are hard to approximate in general, within any polynomial-time computable function of the input length, and this result is valid when the traversal costs are unbounded. When the traversal costs are bounded, we prove that the problems remain $\nph$ even when the underlying network is a star. However, we prove that this pair of problems are polynomial-time solvable in trees when the degrees of the vertices are bounded.

We then prove that the second pair of problems are $\apxh$ under bounded traversal costs in directed graphs. On the positive side, for these problems, we present a polynomial-time algorithm on trees (where both the degrees, and traversal costs are unbounded). We extend this algorithm to graphs where the difference between the number of edges and the number of vertices is constant, i.e. graphs that are in some sense close to trees. However, this extension does not cover cactus graphs, which have treewidth $2$. To solve the problem for other special graph classes such as cactus graphs or bounded treewidth graphs is one possible research direction. Another interesting research direction is to analyze these problems from the perspective of parameterized complexity where a few important and practical parameters are the treewidth of the graph, the number of colors, and the ratio of the largest traversal cost to the smallest non-zero traversal cost.

\section{Preliminaries}
\label{sec:Preliminaries}
\runningtitle{Graphs, trees}
Given an undirected graph $G=(V(G), E(G))$ and a vertex $v \in V(G)$, $\delta_G(v)$ denotes the set of edges incident to $v$ in $G$, and $d_G(v) \defined \abs{\delta_G(v)}$ is the degree of $v$ in $G$. We denote a pair of vertices $u,v \in V(G)$ as $uv$, i.e. $uv \in E(G)$ if $u$ and $v$ are adjacent in $G$. The minimum and maximum degrees of $G$ are defined as $\delta(G)\defined \min \set{d_G(v) | v \in V(G)}$ and $\Delta(G)\defined \max \set{d_G(v) : v \in V(G)}$. Given a tree $T$ and two vertices $v_1,v_2 \in V(T)$, we denote by $P_T(v_1,v_2)$ the path between $v_1$ and $v_2$ in $T$. We denote a bipartite graph as a triple $(V_1, V_2, E)$ where $\set{V_1,V_2}$ is the bipartition of its vertices and $E$ is its edge set. Given two graphs $G$ and $G'$, their union is $G \cup G' \defined (V(G) \cup V(G'), E(G) \cup E(G'))$.

The numbers of the inbound and outbound arcs of a vertex in a digraph are called its in-degree and out-degree, respectively. We denote the ordered pair $(u,v)$ of vertices of $G$ as $uv$, i.e. $uv \in E(G)$ if there is an arc from $u$ to $v$ in $G$. A digraph is a \emph{rooted tree} or \emph{arborescence} if its underlying graph is a tree and it contains a unique \emph{root} vertex denoted by $\root{T}$ which has a directed path to every other vertex of $T$. Each vertex $v \neq \root{T}$ has in-degree $1$. In this case we denote by $\inn{T}{v}$ the unique inbound arc of $v$ in $T$, and by $\parent{T}{v}$ the other endpoint of $\inn{T}{v}$. The set $\children{T}{v}$ is the set of all vertices $u$ of $T$ such that $\parent{T}{u}=v$. We also denote by $P_T(u,v)$ the unique path between vertices $u$ and $v$ in the tree $T$.

\runningtitle{Reload and changeover costs}
We consider proper edge colorings $\chi:E(G) \rightarrow X$ of a given graph $G$ where the colors are taken from a set $X$ and edges incident to the same vertex are assigned different colors. Without loss of generality we assume $X = \set{1,2,\ldots,\abs{X}}$. Since, by Vizing's theorem, every graph is $\Delta(G)+1$ edge colorable, we assume that $\abs{X} \geq \Delta(G)+1$ so that $G$ is edge-colorable with colors from $X$.

The traversal costs are given by a nonnegative function $\tc:X^2 \rightarrow \mathbb{R}^{+} \cup \set{0}$ satisfying
\begin{enumerate}[i)]
\item $\tc(i,j)=\tc(j,i)$ for every $i,j \in X$.
\item $\tc(i,i)=0$ for every $i \in X$.
\end{enumerate}

We denote as $\tc_{max}$ the maximum ratio between two positive traversal costs, i.e., $\tc_{max} \defined \frac{\max \set{\tc_{i,j}|i,j \in X}}{\min \set{\tc_{i,j}|i,j \in X, \tc_{i.j} > 0}}$.
Let $P=(e_1,e_2,\ldots,e_\ell)$ be a path of length $\ell$ of $G$. We denote by $\traversals(P)=\set{(e_i,e_{i+1}): 1 \leq i < \ell}$ the set of traversals of $P$. The traversal cost associated with a traversal $t_i=(e_i, e_{i+1})$ of $P$ with coloring $\chi$ is $\tc_\chi(t_i) \defined \tc(\chi(e_i),\chi(e_{i+1}))$. The traversal cost associated with $P$ is $\tc_\chi(P) \defined \sum_{t \in \traversals(P)} \tc_\chi(t)$. Note that $\tc_\chi(P)=0$ whenever the length of $P$ is zero or one, since in these cases $\traversals(P)=\emptyset$. Therefore, we assume that all paths under consideration have length at least $2$.

Let $\pp$ be a set of paths. The set of traversals of $\pp$ is $\traversals(\pp) \defined \bigcup_{P \in \pp} \traversals(P)$. The \emph{reload cost} of a set of paths $\pp$ with coloring $\chi$ is
\[
\rc_\chi(\pp) \defined \dsum_{P \in \pp} \tc_\chi(P) = \dsum_{P \in \pp} \dsum_{t \in \traversals(P)} \tc_\chi(t),
\]
and its \emph{changeover cost} is
\[
\co_\chi(\pp) \defined \dsum_{t \in \traversals(\pp)} \tc_\chi(t).
\]

Note that the difference between $\rc_\chi(\pp)$ and $\co_\chi(\pp)$ is that if a traversal occurs more than once, it contributes to $\co_\chi(\pp)$ only once, whereas every occurrence contributes to $\rc_\chi(\pp)$.

\runningtitle{Problem Statement}
We assume without loss of generality that $E(G) = \cup_{P \in \pp} E(P)$, i.e. every edge of $G$ is used by at least one path. We note that whenever every traversal is in at most one path of $\pp$, we have $\rc_\chi(\pp)=\co_\chi(\pp)$. Observe that, in particular, this holds when $\pp$ is a set of \emph{distinct} paths with length $2$. This simple fact will be useful throughout this work.

The minimum reload (resp. changeover) cost edge coloring ($\minrc$ resp. $\mincc$) problem aims to find a proper edge coloring of $G$ leading to a minimum reload (resp. changeover) cost with respect to $\pp$.  Formally,\\

\noindent \fbox{\begin{minipage}{\textwidth}
$\minrc/\mincc$ $\instrc$ \\
{\bf Input:} A set of paths $\pp$ comprising a graph $G=\cup \pp$, a set $X$ of at least\\
$\Delta(G)+1$ colors, a traversal cost function $\tc:X^2 \to \mathbb{R}^+ \cup \set{0}$\\
{\bf Output:} A proper edge coloring $\chi:E(G) \rightarrow X$\\
{\bf Objective:} Minimize $\rc_\chi(\pp)/\co_\chi(\pp)$.
\end{minipage}}

Given a tree $T$ and a vertex $r \in V(T)$, let $\pp(T,r) \defined \set{P_T(r,v):v \in V(T) \setminus \set{r}}$ be the set of all paths between the root vertex $r$ and all other vertices. The reload and changeover costs of $T$ rooted at $r$ are $\rc_\chi(T,r) \defined \rc_\chi(\pp(T,r))$ and $\co_\chi(T,r) \defined \co_\chi(\pp(T,r))$, respectively. Given a graph $G$ and a vertex $r$ of $G$, the minimum reload cost path tree edge coloring ($\minrcpt$) and minimum changeover cost arborescence edge coloring ($\mincca$) problems aim to find a proper edge coloring of $G$ and a spanning tree $T$ rooted at $r$ with minimum reload and changeover cost, respectively. Formally,
~\\
\noindent \fbox{\begin{minipage}{\textwidth}
$\minrcpt/\mincca$ $\instrcpt$ \\
{\bf Input:} A graph $G$, a vertex $r$ of $G$, a set $X$ of at least $\Delta(G)+1$ colors, a traversal cost function $\tc:X^2 \to \mathbb{R}^+ \cup \set{0}$\\
{\bf Output:} A proper edge coloring $\chi:E(G) \to X$ and a spanning tree $T$ of $G$ rooted at $r$\\
{\bf Objective:} Minimize $\rc_\chi(T,r)/\co_\chi(T,r)$.
\end{minipage}}
~\\

\runningtitle{Approximation Algorithms, Reductions} Let $\Pi$ be a minimization
problem and $\rho \geq 1$. A (feasible) solution $S$ of an instance
$I$ of $\Pi$ is a \emph{$\rho$-approximation} if the objective
function value of $S$ is at most $\rho$ times the optimum. A polynomial-time algorithm $ALG$ is
a \emph{$\rho$-approximation algorithm} for $\Pi$ if $ALG$ returns a $\rho$-approximation $ALG(I)$
for every instance $I$ of $\Pi$. A \emph{polynomial time approximation scheme} ($\ptas$) for $\Pi$ is an infinite family of algorithms $\set{ALG_\epsilon | \epsilon > 0}$ such that $ALG_\epsilon$ is a $(1+\epsilon)$-approximation algorithm with running time $O(\abs{I}^{h(\epsilon)})$ for some function $h$. $\ptas$ also denotes the class of problems that admit a $\ptas$. $\apx$ is the class of problems that admit a $c$-approximation for some constant $c$.

Given two optimization problems $\Pi$ and $\Pi'$ with objective functions $c_\Pi$ and $c_{\Pi'}$ respectively, an \emph{$L$-reduction} from $\Pi$ to $\Pi'$ consists of two polynomial-time computable functions $f, g$ such that a) $f$ transforms every instance $I$ of $\Pi$ to an instance $f(I)$ of $\Pi'$, b) $g$ transforms every solution $s'$ of $f(I)$ to a solution $g(s')$ of $I$, c) $\abs{OPT_{\Pi'}(f(I))} \leq \rho \cdot \abs{OPT_{\Pi}(I)}$, and $\abs{OPT_\Pi(I)-c_\Pi(g(s'))} \leq \rho' \cdot \abs{OPT_{\Pi'}(f(I)) - c_{\Pi'}(s')}$ for two constants $\rho, \rho'$. A problem is in $\apxh$ if every problem in $\apx$ can be reduced to it by a an $L$-reduction. If a problem is in $\apxh$ then it does not admit a $\ptas$ unless $\nohyp$. In this work, we use a different type of reduction that we will term \emph{$LT$-reduction} or a \emph{Turing type $L$-reduction}. An $LT$-reduction from $\Pi$ to $\Pi'$ is a polynomial-time computable sequence of pairs of functions such that at least one of them is an $L$-reduction from $\Pi$ to $\Pi'$. Clearly, by returning the best solution implied by the individual reductions, one can get a constant approximation to $\Pi$.

\runningtitle{Biconnected Components and Block Trees} A \emph{cut vertex} (\emph{articulation point} or \emph{separation vertex}) of a connected graph is a vertex whose removal (along with its incident edges) disconnects the graph. A graph with no articulation points is \emph{biconnected}. A maximal biconnected induced subgraph of a graph is called a \emph{biconnected component} or a \emph{block} \cite{bondymurty2008graphtheory}. Any connected graph $G$ can be decomposed in linear time into a tree whose vertices are the biconnected components of $G$ and its articulation points. This tree is termed the \emph{block tree} (or \emph{superstructure}) of $G$. The edges of the block tree join every cut vertex to the blocks it belongs to \cite{E79} and every block to the cut vertices (of $G$) contained in it.

\runningtitle{Matchings} A matching of a graph $G$ is a subset $M \subseteq E(G)$ of pairwise nonadjacent edges. A matching is perfect if $V(M)=V(G)$. In an edge-weighted graph, minimum weight perfect matching is a perfect matching with minimum total edge weight and it can be computed in polynomial time.

\runningtitle{The $\lightksubgraph$ Problem} The $\lightksubgraph$ problem is to find an induced subgraph $H$ on $k$ vertices, of a given edge-weighted graph $G$, with minimum total edge weight. This problem is $\nph$ in the strong sense even when the graph is a complete graph and the edge weights are either $1$ or $2$ \cite{watrigant2012k}.

\runningtitle{The Minimum Set Cover Problem ($\minsc$)} An instance of this problem is a set system $\cs=\set{S_1, S_2, \ldots, S_m}$, with $U \defined \cup \cs$. Given such an instance, one has to find a subset $\cc \subseteq \cs$ that covers $U$, i.e. $\cup \cc = U$, such that $\abs{\cc}$ is minimum. The special case in which $\forall i, |S_i| \leq k$, and $\forall u \in U, \abs{\set{S_i \in \cs: u \in S_i}} \leq \ell$, for two constants $k,\ell>0$ is called the minimum $(k,\ell)$-set cover problem. The minimum $(3,2)$-set cover problem ($\minthreesctwo$) is $\apxh$ \cite{duh1997approximation}, i.e. it does not admit a $\ptas$ unless $\nohyp$.

\section{Hardness Results}\label{sec:Herdness}
In this section we show that $\mincc$ and $\minrc$ are inapproximable when the ratio $\kappa_{\tc}$ of the biggest traversal cost to the smallest non-zero traversal cost is unbounded. Then, we show that both problems remain $\nph$ in the strong sense even when $\kappa_{\tc}=2$ and $G$ is a star. We then return to the $\mincca$ and $\minrcpt$ problems and show that they are $\apxh$ in directed graphs even when $\kappa_{\tc}=2$.

\begin{theorem}
$\mincc$ and $\minrc$ are inapproximable within any polynomial-time computable function $f(\abs{\pp})$.
\end{theorem}
\begin{proof}
The proof is by reduction from the chromatic index problem. The chromatic index of a graph $G$ is either $\Delta(G)$ or $\Delta(G)+1$. However, it is $\npc$ to decide between these two values \cite{holyer1981np}. Given a graph $G$ we construct an instance $I=\instrc$ where $\pp$ consists of all distinct paths of length $2$ of $G$, $\abs{X}=\Delta(G)+1$, and
\[
\tc(i,j)=\left\{
\begin{array}{ll}
0 & \textrm{if~}i=j\\
1 & \textrm{if~}i \neq j \textrm{~and~} i,j \leq \Delta(G)\\
M & \textrm{otherwise}
\end{array}
\right.
\]

\newcommand{\upperb}{\abs{\pp} \cdot f(\abs{\pp})}

\noindent where $M=\upperb$. We recall that since the paths of $\pp$ are of length $2$, we have $\rc_\chi(\pp)=\co_\chi(\pp)$ for every coloring $\chi$. Assume, by way of contradiction, that there exists an $f(\abs{\pp})$-approximation algorithm $\ca$ for one of the problems. Then $\ca$ is an $f(\abs{\pp})$-approximation algorithm for both problems. If the chromatic index of $G$ is $\Delta(G)$, let $\chi$ be a proper edge coloring of $G$ using the first $\Delta(G)$ colors. Then all traversal costs are $1$, and $\rc_\chi(\pp)=\co_\chi(\pp)=\abs{\pp}$; therefore, the solution has value at most $\upperb$. On the other hand, if the chromatic index of $G$ is $\Delta(G)+1$, then any edge coloring $\chi'$ uses $\Delta(G)+1$ colors, and we have $\rc_{\chi'}(\pp)=\co_{\chi'}(\pp) \geq \abs{\pp}+M-1=\abs{\pp}+\upperb-1>\upperb$ since there is at least one traversal with cost $M$. Therefore, $G$ is $\Delta(G)$ edge-colorable if and only if $\ca$ returns a solution with cost at most $\upperb$.
\qed
\end{proof}

We now show that both problems are $\nph$ in the strong sense even in very simple graphs that are in particular $\Delta(G)$- edge-colorable, namely stars, and have $\kappa_{\tc}=2$.
\begin{theorem}\label{thm:HardnessStarRhoTwo}
$\mincc$ and $\minrc$ are $\nph$ in the strong sense even when $\tc(i, j) \in \set{0,1,2}$ for every pair $i,j \in X$ and $G$ is a star.
\end{theorem}
\begin{proof}
The proof is by reduction from the $\lightksubgraph$ problem, which is $\nph$ in the strong sense even on complete graphs with edge weights either $1$ or $2$ \cite{watrigant2012k}. Given such an instance $(K,w)$ of~\lightksubgraph~where $K$ is a complete graph on more than $k$ vertices and $w$ is the edge weight function such that $w_{ij}$ is the weight of the edge between vertices $i$ and $j$, we build the following instance: $G$ is a star on $k+1$ vertices ($k$ leaves), $\pp$ consists of the $\comb{k}{2}$ paths between every pair of leaves of $G$, $\abs{X}=\abs{K}$, and $\tc(i,j)=w_{ij}$. Since all paths have length $2$, we have $\rc_\chi(\pp)=\co_\chi(\pp)$ for every coloring $\chi$. Moreover, $\rc_\chi(\pp)$ is equal to the total edge weight of a clique on $k$ vertices of $K$ (corresponding to the set of $k$ colors of $X$ used in $\chi$).
\qed
\end{proof}

\begin{figure} [htbp]
\begin{center}
\commentfig{
\includegraphics [width=\textwidth]{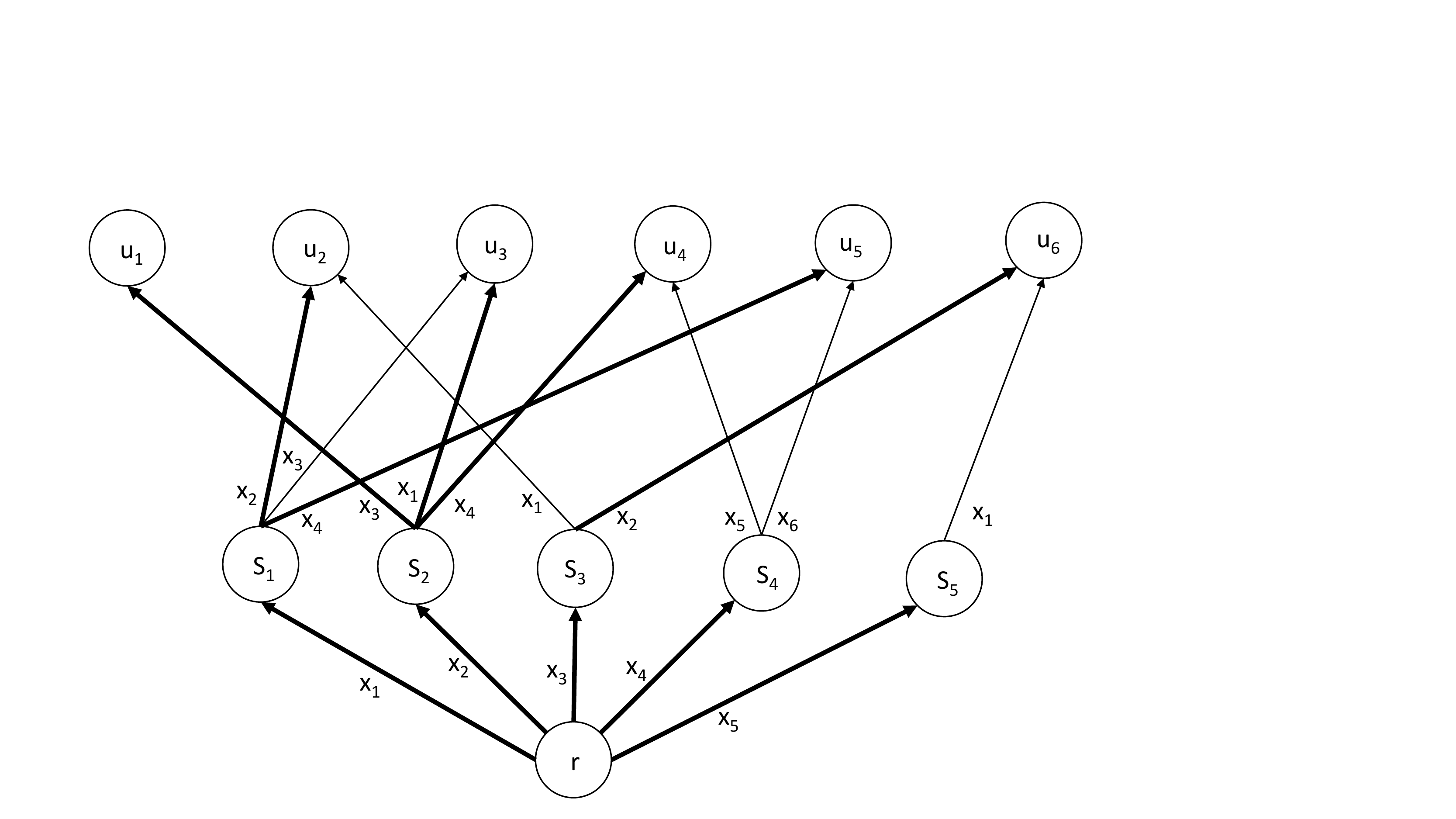}
}
\caption{The directed acyclic graph $G$ corresponding to an instance of $\mincca$ with $S_{1}=\{u_{2},u_{3}, u_{5}\}$, $S_{2}=\{u_{1}, u_{3}, u_{4}\}$, $S_{3}=\{u_{2}, u_{6}\}$, $S_{4}=\{u_{4}, u_{5}\}$, $S_{5}=\{u_{6}\}$, $X_{c}=\{x_1, x_2, x_3, x_4\}$, and $X_e=\{x_5, x_6\}$. Bold arcs indicate the spanning tree corresponding to a minimum set cover $\mathcal{C}^*=\{S_{1},S_{2}, S_{3} \}$.}
\label{fig:minccaecreduction}
\end{center}
\end{figure}

\begin{theorem}
$\mincca$ and $\minrcpt$ are $\apxh$ in directed graphs even when $\tc(i, j) \in \set{0,1,2}$ for every pair $i,j \in X$.
\end{theorem}

\begin{proof}
The proof is by $LT$-reduction from $\minthreesctwo$, which is known to be $\apxh$. We consider an instance $\cs$ of $\minthreesctwo$ with $n$ elements, $m \geq 4$ sets and with every set cover consisting of at least 4 sets. Otherwise, an optimal set cover can be found in polynomial time.
Given such an instance and an integer $k \leq m$, we construct an instance $f_k(\cs)=\minccainst$ as follows (see Figure \ref{fig:minccaecreduction}). $G=(V,E)$ is a directed acyclic graph with $V=\set{r} \cup \cs \cup U$ where $U = \cup \cs$, $E=E_1 \cup E_2$, $E_1 = \set{r S_i|~S_i \in \cs}$ and $E_2 = \set{S_i u_j|~S_i \in \cs, u_j \in S_i}$. The color set is the disjoint union $X$ of two color sets $X_c$ and $X_e$ with $\abs{X_c}=k$ and $\abs{X_e}=\Delta(G)+1-k$. Finally,
\[
\tc(x,y) = \left\{
\begin{array}{ll}
0 & \textrm{if~} x=y\\
1 & \textrm{if~} x \neq y \textrm{~and~} x,y \in X_c\\
2 & \textrm{otherwise}.
\end{array}
\right.
\]

We observe that $\rc_\chi (T,r) = \co_\chi (T,r)$ since every directed path has length at most $2$. Therefore, our reduction behaves in the same way for $\mincca$ and $\minrcpt$. For any feasible solution $(T,\chi)$ of $f_k(\cs)$, the set of parents of the vertices $U$ in $T$ is a set cover $g_k(T, \chi)$. This completes the description of the reduction.

Let $\cc^*$ be a minimum set cover of $\cs$ and $k^* = \abs{\cc^*}$. We now show that $f_{k^*}, g_{k^*}$ is an $L$-reduction.
Consider the subgraph of $G$ induced by the vertices $\set{r} \cup \cc^* \cup \set{r}$.
All the vertices of this graph have degree at most 4, except $r$ whose degree is $k^* \geq 4$.
Then this subgraph is bipartite with maximum degree $k^*$.
Therefore, we can color all the arcs of this subgraph with the $k^*$ colors of $X_c$, and the remaining arcs (all incident to $r$) with colors from $X_e$.
The cost of this solution is $n$, thus $OPT(f_{k^*}(\cs)) \leq n$.
Finally, a spanning tree $T$ is built by joining every vertex $u_i$ to an arbitrary vertex $S_j \in \cc^*$ such that $u_i \in S_j$ and each $u_i$ is a leaf of the spanning tree. We conclude that
\begin{equation}
OPT(f_{k^*}(\cs)) \leq n \leq 3 \cdot OPT(\cs) \label{eqn:LReductionOne}
\end{equation}
where the last inequality holds since every set can cover at most three elements.

Let $(T, \chi)$ be a solution of $f_{k^*}(\cs)$. We partition $g_{k^*}(T,\chi)$ into two sets $\cs_c = \{S_i \in g_{k^*}(T,\chi)|~ \chi(r S_i) \in X_c \}$ and $\cs_e= \{S_i \in g_{k^*}(T,\chi)|~ \chi(r S_i) \in X_e\}$. We observe that $\abs{\cs_c} \leq k^*$ since $\chi$ colors the inbound arcs of $\cs_c$ with distinct colors of $X_c$ (all these arcs are incident to $r$) and $\abs{X_c}=k^*$. Therefore,
\[
\abs{g_{k^*}(T,\chi)} - OPT(\cs) =  \abs{g_{k^*}(T,\chi)} - k^* = \abs{\cs_e} + \abs{\cs_c} - k^* \leq \abs{\cs_e}.
\]

We have that $\co_\chi (T,r)\geq n+\abs{\cs_e}$ since the arc leading to $u_i$ in $T$ incurs a traversal cost of at least $1$ in the parent $S_j$ of $u_i$ and for every $S_j \in \cs_e$ there is at least one traversal that costs $2$. Therefore,
\[
\abs{\cs_e} \leq \co_\chi (T,r) - n \leq \co_\chi (T,r) - OPT(f_{k^*}(\cs)).
\]
We combine the last two inequalities to get
\begin{equation}
\abs{g_{k^*}(T,\chi)} - OPT(\cs) \leq \co_\chi (T,r) - OPT(f_{k^*}(\cs)). \label{eqn:LReductionTwo}
\end{equation}
By inequalities (\ref{eqn:LReductionOne}) and (\ref{eqn:LReductionTwo}), $f_{k^*}$ and $g_{k^*}$ constitute an $L$-reduction, as required.
\qed
\end{proof}

\section{Polynomial-time Solvable Cases}\label{sec:Algorithms}
In this section we present polynomial-time algorithms for some special cases. We start with the definitions and notations used in this section. We denote the set of all proper edge colorings of a graph $G$ by $\ff_G$. Two partial functions $f,f'$ \emph{agree} if $f(x)=f'(x)$ whenever both $f$ and $f'$ are defined on $x$. We denote this fact by $f \sim f'$.

The following discussion refers to the changeover cost; however, it holds for the reload cost, too. Whenever this is not the case, the difference between the two costs will be made explicit. For a subgraph $H$ of $G$ we say that a traversal $(e_i,e_j)$ is \emph{within} $H$ if $e_i,e_j \in E(H)$, and we denote by $\traversals(\pp,H)$ the set of traversals of $\pp$ within $H$. We define $\rc_\chi(\pp,H)$ and $\co_\chi(\pp,H)$ similarly by taking into account only traversals within $H$. Let $H_1,H_2,\ldots,H_k$ be subgraphs of $G$ such that every traversal is within exactly one subgraph $H_i$. Clearly, $\co_\chi(\pp)=\sum_{i=1}^k \co_\chi(\pp,H_k)$.

\begin{figure} [htbp]
\begin{center}
\commentfig{
\includegraphics [scale=0.4]{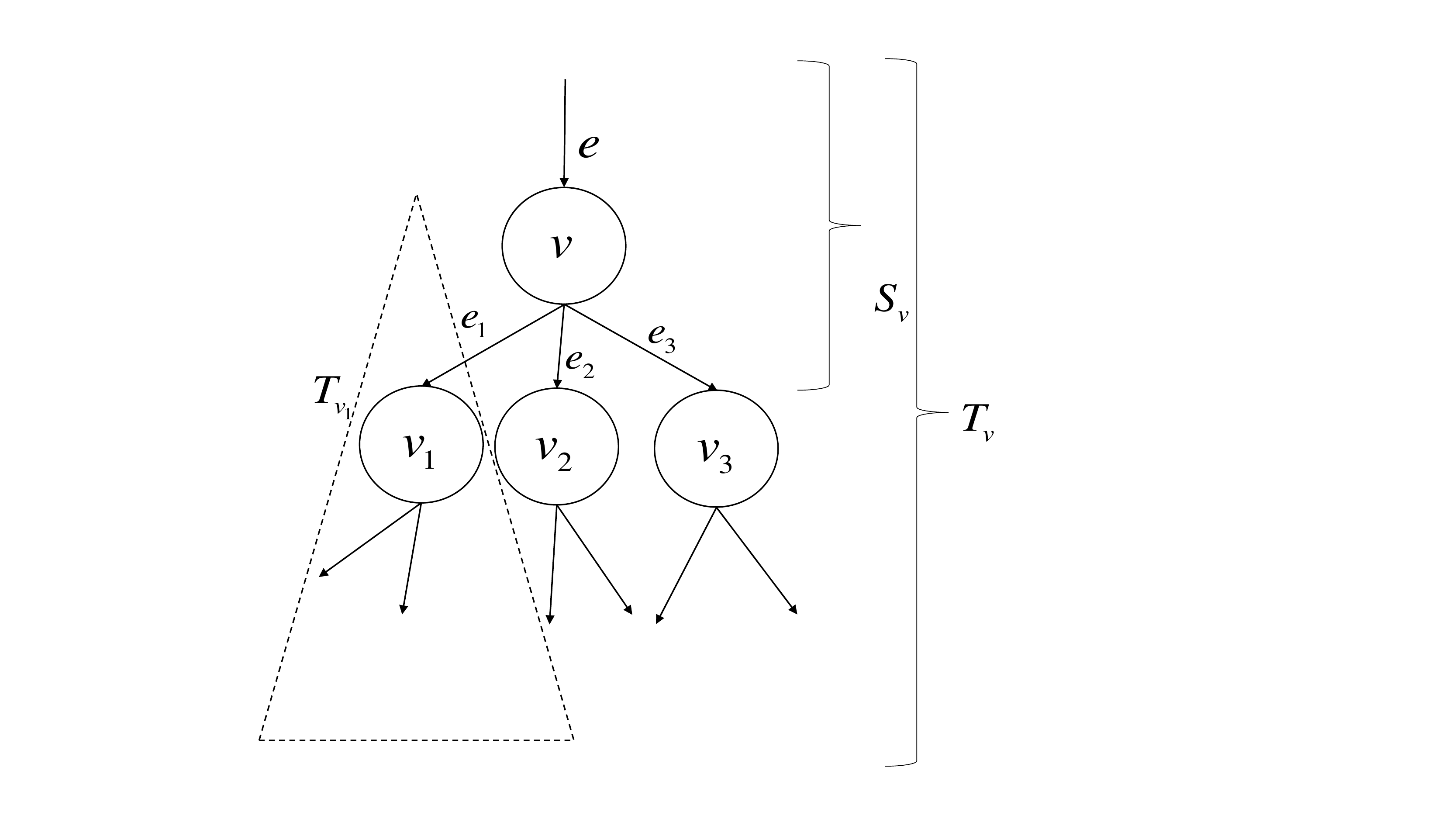}
}
\caption{Notation for Section \ref{sec:Algorithms}.}
\label{fig:treedynamicprog}
\end{center}
\end{figure}

In the sequel, we analyze the decomposition of a spanning tree $T$ of $G$ (and its cost) by the block tree of $G$.
Let $T$ be a spanning tree of $G$ rooted at some vertex $r$. For the $\minrcpt$ and $\mincca$ problems, it is convenient to choose the root as the vertex $r$ given in the instance. Consult Figure \ref{fig:treedynamicprog} for the following discussion.
For a non-root vertex $v$ of $T$, we denote by $T_v$ the subtree of $T$ rooted at $v$ with the addition of the arc $e=\inn{T}{v}$. Let $S_v$ denote the star consisting of $v$ and its incident edges. Moreover, let $\children{T}{v}=\set{v_1,\ldots,v_k}$ and $e_i = vv_i$. Clearly, every traversal within $T_v$ is either within $S_v$ or within $T_{v_i}$ for some $i \in [k]$. Therefore,
\begin{equation}
\co_\chi(\pp,T_v) = \co_\chi(\pp,S_v) + \sum_{i=1}^k \co_\chi(\pp,T_{v_i}).\label{eqn:coSubtreeDecomposition}
\end{equation}
We denote by $OPT_\co(\pp, v, x)$ the minimum changeover cost within $T_v$, among all colorings $\chi$ such that $\chi(\inn{T}{v})=x$. Formally,
\begin{equation}
OPT_\co(\pp, v, x) \defined \min \set{\co_\chi(\pp, T_v):\chi \in \ff_{T_v}, \chi(\inn{T}{v})=x}.\label{eqn:OPTccDefinition}
\end{equation}
Since $T_v$ does not contain any traversals whenever $v$ is a leaf, we have
\begin{equation}
OPT_\co(\pp, v, x) = 0 \textrm{~whenever $v$ is a leaf}.\label{eqn:DynProgInitial}
\end{equation}

In order to compute $OPT_\co(\pp, v, x)$ we categorize all proper edge colorings $\chi$ of $T_v$ by the colorings $\chi_v$ they induce on $S_v$:\\
\begin{eqnarray}
\alpha_\co(\chi_v)  \defined  \min & \set{\co_\chi(\pp, T_v): \chi \in \ff_{T_v}, \chi \sim \chi_v},\nonumber \\
OPT_\co(\pp, v, x) = & \min \set{\alpha_\co(\chi_v) : \chi_v \in \ff_{S_v}, \chi_v(e)=x}.\label{eqn:OptPpEx}
\end{eqnarray}
where $\alpha_\co(\chi_v)$ is the minimum cost  within $T_v$ among all colorings that induce the coloring $\chi_v$ on $S_v$. We define $\alpha_\rc$ similarly.
For a fixed coloring $\chi_v \in \ff_{S_v}$, we proceed as follows by first using (\ref{eqn:coSubtreeDecomposition}):
\begin{eqnarray}
\alpha_\co(\chi_v) & = & \min \set{\co_\chi(\pp,S_v) + \sum_{i=1}^k \co_\chi(\pp,T_{v_i}) : \chi \in \ff_{T_v}, \chi \sim \chi_v} \nonumber\\
& = & \co_{\chi_v}(\pp,S_v) + \sum_{i=1}^k  \min \set{ \co_\chi(\pp,T_{v_i}): \chi \in \ff_{T_v}, \chi \sim \chi_v} \nonumber\\
& = & \co_{\chi_v}(\pp,S_v) + \sum_{i=1}^k  \min \set{ \co_\chi(\pp,T_{v_i}): \chi \in \ff_{T_v}, \chi(e_i) = \chi_v(e_i)} \nonumber\\
& = & \co_{\chi_v}(\pp,S_v) + \sum_{i=1}^k  OPT_\co(\pp, v_i, \chi_v(e_i)) \nonumber \\
& = & \co_{\chi_v}(\pp,S_v) + \sum_{v' \in \children{T}{v}}  OPT_\co(\pp, v', \chi_v(vv')), \label{eqn:alpha}
\end{eqnarray}
where the equality second to last holds by (\ref{eqn:OPTccDefinition}). In particular, for $v=r$ we obtain the optimum of the instance as:
\begin{equation}
\co^*(\pp) = \min \set{\alpha_\co(\chi_r) : \chi_r \in \ff_{S_r}}.\label{eqn:DynProgFinal}
\end{equation}

Equations (\ref{eqn:DynProgInitial}) through (\ref{eqn:DynProgFinal}) imply Algorithm \ref{alg:DynProgTrees}, which is a dynamic programming algorithm for the case when $G$ is a tree.
The number of entries $OPT_\co(\pp, v, x)$ computed by the algorithm is $\abs{V(G)} \cdot \abs{X}$, i.e. polynomial in the size of the input. In order to get a polynomial-time algorithm, we have to compute every entry in polynomial time. The value $\co_{\chi_v}(\pp,S_v)$ can be computed in time $O(\abs{\pp})$ since every path has at most one traversal in $S_v$. Therefore, $\alpha_\co(\chi_v)$ can be computed in time $O(\abs{\pp}+k)=O(\abs{\pp}+\Delta(G))$. In the sequel, we analyze the computation time of $OPT_\co(\pp, v, x)$ in various cases.

\alglanguage{pseudocode}

\begin{algorithm}[H]
\caption{Dynamic Programming for $\mincc$ in Trees}\label{alg:DynProgTrees}
\begin{algorithmic}
\Require {Input $\instrc$}
\Require {$\cup \pp$ is a tree}
\State $T \gets \cup \pp$
\State Root $T$ at an arbitrary vertex $r$
\ForAll {vertex $v \in V(T) \setminus \set{r}$ in a postorder traversal of $T$}
\ForAll {$x \in X$}
\If {$v$ is a leaf of $T$}
\State $OPT_\co(\pp,v,x) \gets 0$
\Else
\State $OPT_\co(\pp,v,x) \gets $ \Call{ComputeNonLeaf}{$\pp,v,x$}
\EndIf
\EndFor
\EndFor
\State \Return $\co^*(\pp)$ using Equation (\ref{eqn:DynProgFinal})
\Statex
\Function {ComputeNonLeaf}{$\pp,v,x$}
\State $e \gets \inn{T}{v}$
\State $min \gets \infty$
\For {$\chi_v \in \ff_{S_v} \wedge \chi_v(e)=x$}
\State $\co_{\chi_v(\pp,S_v)} \gets \sum_{t \in \traversals(\pp,S_v)} \tc{\chi_v}(t)$
\State Compute $\alpha_\co(\chi_v)$ using Equation (\ref{eqn:alpha})
\If {$\alpha_\co(\chi_v) < min$}
\State $min \gets \alpha_\co(\chi_v)$
\EndIf
\EndFor
\State \Return $min$
\EndFunction
\end{algorithmic}
\end{algorithm}

\begin{theorem}
$\mincc$ and $\minrc$ are solvable in polynomial time when $G$ is a tree and $\abs{X}^{\Delta(G)}$ is polynomial in the input size.
\end{theorem}
\begin{proof}
The number of proper edge colorings $\chi_v$ of $S_v$ using colors from $X$ such that $\chi_v(e)=x$ is at most $\abs{X}^{\Delta(G)}$. Therefore, the computation time of $OPT_\co(\pp, v, x)$ (in Function \textsc{ComputeNonLeaf}) is at most $O((\abs{\pp}+\Delta(G))\abs{X}^{\Delta(G)})$, i.e. polynomial in the input size.
\qed
\end{proof}
\begin{corollary}\label{coro:PolynomialCases}
$\mincc$ and $\minrc$ are solvable in polynomial time whenever $G$ is a bounded degree tree.
\end{corollary}

Note that Corollary \ref{coro:PolynomialCases} complements Theorem \ref{thm:HardnessStarRhoTwo}, which proves that $\mincc$ and $\minrc$ are $\nph$ for unbounded degree stars. In the sequel, we show that $\mincc$ and $\minrc$ are polynomial-time solvable in trees and graphs having a structure close to a tree, i.e. graphs $G$ where $\abs{E(G)}-\abs{V(G)}$ is bounded by some constant.

\begin{theorem}\label{thm:MinCCECsinglesource}
$\mincc$ and $\minrc$ are solvable in polynomial time when $G$ is a tree and a particular vertex $r$ is an endpoint of every path $P \in \pp$.
\end{theorem}
\begin{proof}
We consider $G$ as rooted at $r$. We observe that since all paths have an endpoint at $r$, all traversals within $S_v$ contain the edge $e=\inn{T}{v}$. Therefore, for $\chi_v(e)=x$,
\[
\co_{\chi_v}(\pp,S_v) = \sum_{i=1}^k \tc(\chi_v(e),\chi_v(e_i)) = \sum_{i=1}^k \tc(x,\chi_v(e_i))
\]
and
\[
\rc_{\chi_v}(\pp,S_v) = \sum_{i=1}^k \tc(x,\chi_v(e_i)) \abs{\pp_{e_i}}
\]
where $\pp_{e_i}$ is the set of paths of $\pp$ that contain $e_i$. Substituting in (\ref{eqn:alpha}) we get
\[
\alpha_\co(\chi_v)= \sum_{i=1}^k  \left( \tc(x,\chi_v(e_i)) +  OPT_\co(\pp, v_i, \chi_v(e_i)) \right),
\]
and similarly
\[
\alpha_\rc(\chi_v)= \sum_{i=1}^k  \left( \tc(x,\chi_v(e_i)) \abs{\pp_{e_i}} +  OPT_\rc(\pp, v_i, \chi_v(e_i)) \right).
\]
We now observe that these values can be computed in polynomial time by Function \textsc{ComputeNonleaf} in Algorithm \ref{alg:MinPerfectMatching}, as described in the sequel. Consider the complete bipartite graph $B$ where the bipartition of the vertices is $\set{\set{v_1,\ldots,v_k},X-x}$. There is a one-to-one correspondence between the proper edge colorings $\chi_v$ of $S_v$ with $\chi_v(e)=x$ and the matchings of $B$ with size $k$. The matching $M_{\chi_v}$ corresponding to the edge coloring $\chi_v$ is such that $v_i y \in M_{\chi_v}$ if and only if $\chi_v(e_i)=y$. We assign the weight $w(v_i y)= \tc(x,y) +  OPT_\co(\pp, v_i, y)$ to the edge $v_i y$ for every $i \in [k]$ and $y \in X - x$. Under this setting, the total weight of the matching $M_{\chi_v}$ corresponding to $\chi_v$ is equal to $\alpha_\co(\chi_v)$. We conclude that minimizing $\alpha_\co(\chi_v)$ is equivalent to finding a minimum weight matching with $k$ edges on this weighted graph, i.e. $OPT_\co(\pp, v, x)$ can be computed in polynomial time.
\qed
\end{proof}

\alglanguage{pseudocode}

\begin{algorithm}[H]
\caption{\textsc{ComputeNonLeaf} by Minimum Weight Perfect Matching}\label{alg:MinPerfectMatching}
\begin{algorithmic}
\Require {$\cup \pp$ is a tree}
\Require {All paths in $\pp$ have a common endpoint}
\Function {ComputeNonLeaf}{$\pp,v,x$}
\State $e \gets \inn{T}{v}$
\State Construct a complete bipartite graph $B$ with bipartition $\set{\children{T}{v},X \setminus \set{x}}$
\For {$v' \in \children{T}{v}$}
\For {$y \in X - x$}
\State $w(v'y) \gets \tc(x,y) + OPT_\co(\pp,v',y)$
\EndFor
\EndFor
\State \Return The minimum weight of a perfect matching of $B$.
\EndFunction
\end{algorithmic}
\end{algorithm}

We note that in the special case of $\mincca$ (resp. $\minrcpt$) problem when $G$ is a tree, there is only one spanning tree; therefore, we get a special case of the $\mincc$ (resp. $\minrc$) problem when $G$ is a tree and a particular vertex $r$ of $G$ is the source vertex of all paths. Therefore,
\begin{corollary}
$\mincca$ and $\minrcpt$ are solvable in polynomial time for trees.
\end{corollary}

Consider an input graph $G$ that is not a tree, but a tree can be obtained by the removal of at most $c$ edges from $G$ where $c$ is a constant. Then, one can try all of the at most $\abs{E}^c$ combinations of the edges to be removed, solve the problem for each remaining tree in turn, and return the best solution. Therefore,

\begin{corollary}
$\mincca$ and $\minrcpt$ are solvable in polynomial time for graphs $G$ where $\abs{E(G)}-\abs{V(G)}$ is bounded by some constant.
\end{corollary}

\begin{figure} [htbp]
\begin{center}
\commentfig{
\includegraphics [scale=0.4]{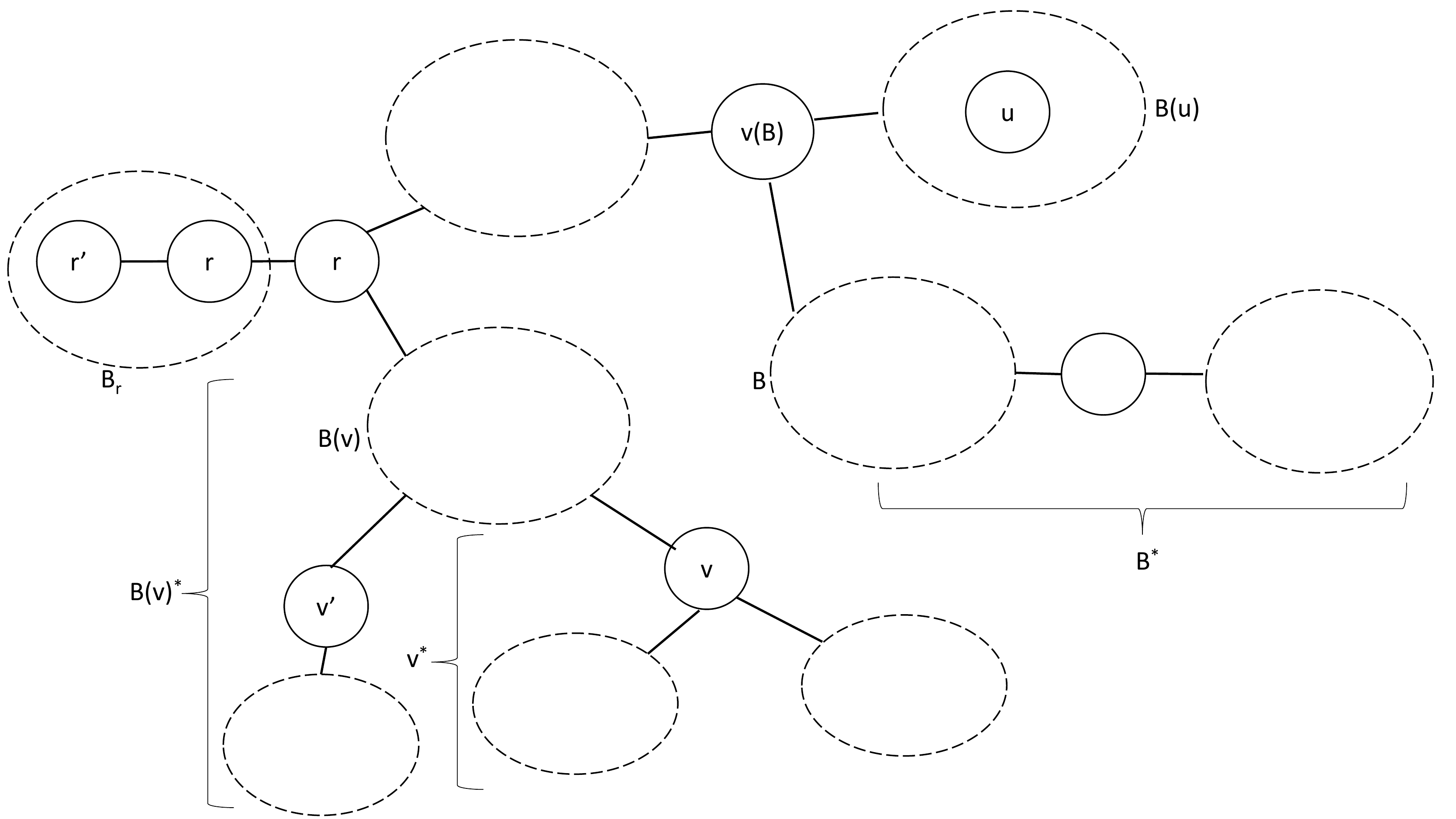}
}
\caption{Notation regarding the vertices of a block tree.}
\label{fig:BlockTreeNotation}
\end{center}
\end{figure}

The following theorem extends this idea for graphs with sparse blocks and bounded degree cut vertices.
\begin{theorem}\label{thm:MINCCADynamicProg}
$\mincca$ and $\minrcpt$ are solvable in polynomial time whenever
\begin{enumerate}[a)]
\item the degree of every cut vertex of $G$ is bounded by some constant $c_1$, and
\item $\abs{E(B)-\abs{V(B)}}$ is bounded by some constant $c_2$ for every block $B$ of $G$.
\end{enumerate}
\end{theorem}
\begin{proof}
\runningtitle{Block Tree Notation}
Consult Figure \ref{fig:BlockTreeNotation} for the definitions and notation we introduce in this paragraph. For simplicity, we modify the instance as follows. We add to $G$ an additional vertex $r'$ incident only to $r$, and an additional color $0 \notin X$ that is usable only in the new edge $rr'$, and $\tc(0,i)=0$ for every $i \in X$. Moreover, we set $r'$ as the root of the modified instance. Clearly, every solution of the new instance contains $rr'$ and there exists an optimal solution in which $rr'$ is colored $0$. After this modification, $r$ is a cut vertex of $G$ and $B_r = \set{r,r'}$ is a block of $G$. We consider the block tree $\TT$ of $G$ as rooted at $B_r$. We recall that the neighbours of a cut vertex (resp. block) in $\TT$ are blocks (resp. cut vertices). For a cut vertex $v$ (resp. block $B$) of $G$, we denote by $B(v)$ (resp. $v(B)$) the parent of $v$ (resp. $B$) in $\TT$. For a non-cut vertex $u$ of $G$, we denote by $B(u)$ the unique block that contains $u$. For a block $B$ (resp. cut vertex $v$), we denote by $B^*$ (resp. $v^*$) the set of all vertices of $G$ contained in the blocks of the subtree of $\TT$ rooted at $B$ (resp. $v$). Note that $B(v) \cup v^* \subseteq B(v)^*$.

\runningtitle{Spanning Trees and the Block Tree}
For a set $U$ of vertices of $G$, we denote by $\cs(U)$ the set of subgraphs of $G$ that are trees and span $U$, i.e. the set of trees $(U,E_T)$ such that $E_T \subseteq E(G)$. Note that possibly $\cs(U)=\emptyset$. In particular, $\cs(G) \defined \cs(V(G))$ is the set of spanning trees of $G$. Let $T$ be a spanning tree of $G$ rooted at $r'$. We note that graph $T[B]$ induced by $T$ on a block $B$ of $G$ is a spanning tree of $B$ rooted at $v(B)$. For a spanning tree $T$ and a cut vertex $v$ of $G$, we denote by $T_{v^*}$ the subtree of $T$ that contains the vertices $v^*$ and the parent of $v$, i.e. $T_{v^*}=T[v^* \cup \parent{T}{v}]$. Note that $T_{v^*}$ is a subtree of $T_v$, but possibly different from $T_v$ since $v$ might have descendants in $B(v)$ and such vertices are not in $v^*$.

\runningtitle{Structure of the Dynamic Programming Algorithm}
We now present a dynamic programming algorithm for the problem (see pseudocode in Algorithm \ref{alg:SparseBlocks}). We first introduce the values to be computed by the algorithm. For every cut vertex $v$ of $G$ and a color $x \in X$, the algorithm computes $OPT_\co(v,x)$ that denotes the minimum changeover cost within $T_{v^*}$ of a spanning tree $T$ of $G$ and a coloring $\chi$ such that $\chi(\inn{T}{v})=x$, i.e.
\[
OPT_\co(v,x) \defined \min \set{\co_\chi(T_{v^*}): T \in \cs(G), \chi(\inn{T}{v})=x}.
\]
Clearly, the optimum of the instance is $OPT_\co(r,0)$.

Given a block $B$ and a spanning tree $\hat{T} \in \cs(B)$ of it, we consider $\hat{T}$ as rooted at $v(B)$. For every such block $B$, spanning tree $\hat{T}$, and every non-root vertex of $\hat{T}$ (i.e. every vertex $v \in B \setminus \set{v(B)}$) the algorithm computes the value $OPT_\co(v,\hat{T},x)$, which denotes the minimum changeover cost within $T_v$ of a spanning tree $T$ of $G$ that induces the tree $\hat{T}$ on $B(v)$ and a coloring $\chi$ such that $\chi(\inn{T}{v})=\chi(\inn{\hat{T}}{v})=x$, i.e.
\[
OPT_\co(v,\hat{T},x) \defined \min \set{\co_\chi(T_v): T \in \cs(G), T[B(v)] = \hat{T} ,\chi(\inn{\hat{T}}{v})=x}
\]
for every $v \in V(B) \setminus \set{v(B)}$, $\hat{T} \in \cs(B)$, and $x \in X \setminus \set{0}$.
A spanning tree of $B$ is obtained by removing $\abs{E(B)}-\abs{V(B)}+1 \leq c_2+1$ edges from $E(B)$. Therefore, $\abs{\cs(B)} \leq \comb{\abs{E(B)}}{c_2+1} \leq \abs{E(B)}^{c_2+1}$. Therefore, the total number of values to be computed is $\bigoh(\abs{V(G)} \cdot \abs{X} \cdot \abs{E(G)}^{c_2+1})$, which is polynomial in the input size. It remains to show a) how to compute each value in time polynomial in the input size, and b) how to compute the optimum once these values are computed.

We perform a bottom-up traversal of $\TT$ during which we compute, at a block $B$, the values $OPT_\co(v,\hat{T},x)$, and at a cut vertex $v$, the values $OPT_\co(v,x)$.

\runningtitle{Algorithm for a block}
We start with the description of the computation at a block $B$. Consult Figure \ref{fig:AVertexWithinABlock} for this description.
Let $B$ be a block of $G$, $\hat{T} \in \cs(B)$, and $T$ a spanning tree of $G$ such that $T[B]=\hat{T}$. We perform a bottom-up traversal of $\hat{T}$. At each non-root vertex $v$ of $\hat{T}$, (i.e. $v \in B - v(B)$) we proceed as follows. Let $e=\inn{T}{v}=\inn{\hat{T}}{v}$, $\children{T}{v}=\set{v_1,\ldots,v_k}$, and $e_i=vv_i$ for $i \in [k]$. We first assume, for simplicity, that $v$ is not a cut vertex. In this case, $\children{T}{v} \subseteq B - v(B)$, and we can compute $OPT_\co(v,\hat{T},x)$, in the same way that we did in Theorem~\ref{thm:MinCCECsinglesource}, namely by reducing the problem to finding a minimum weight perfect matching $\gamma(H,w)$ of the bipartite graph $H$ with bipartition $\set{X-x,\set{e_1,\ldots,e_k}}$ and weights $w(e_i y)=\tc(x,y)+OPT_\co(v_i,\hat{T},y)$. Now assume that $v$ is a cut vertex, and let $\set{v_1,\ldots,v_{k'}}$ be the set of its children in $B$ where $k' < k$. In this case, we divide $T_v$ into the subtrees $T_{v_1},\ldots,T_{v_{k'}}$ and $T_{v^*}$. We note that all these trees intersect exactly at $v$ and that $\co_\chi(T_v)$ is the sum of the costs of these trees with the addition of the traversal costs from $e$ to each $e_i$. Therefore, the optimum cost can be computed by adding $OPT_\co(v,x)$ to $\gamma(H,w)$. Summarizing,
\[
OPT_\co(v,\hat{T},x)= \gamma(H,w)+
\left\{\begin{array}{ll}
OPT_\co(v,x) & \textrm{if~}v\textrm{~is a cut vertex}\\
0 & \textrm{otherwise},
\end{array}
\right.
\]
where $\gamma(H,w)$ is the minimum weight perfect matching of the complete bipartite graph $H$ with bipartition $\set{X-x, \children{\hat{T}}{v}}$ and edge weights
\[
w(e_i y)=\tc(x,y)+OPT_\co(v_i,\hat{T},y).
\]
At this point we note that for the $\minrcpt$ problem the weights of $H$ are
$w(e_i y)=\tc(x,y) \cdot \abs{\pp_{e_i}}+OPT_\rc(v_i,\hat{T},y)$.
This computation can be clearly carried out in polynomial time as in the case of Theorem~\ref{thm:MinCCECsinglesource}.

\begin{figure} [htbp]
\begin{center}
\commentfig{
\includegraphics [scale=0.4]{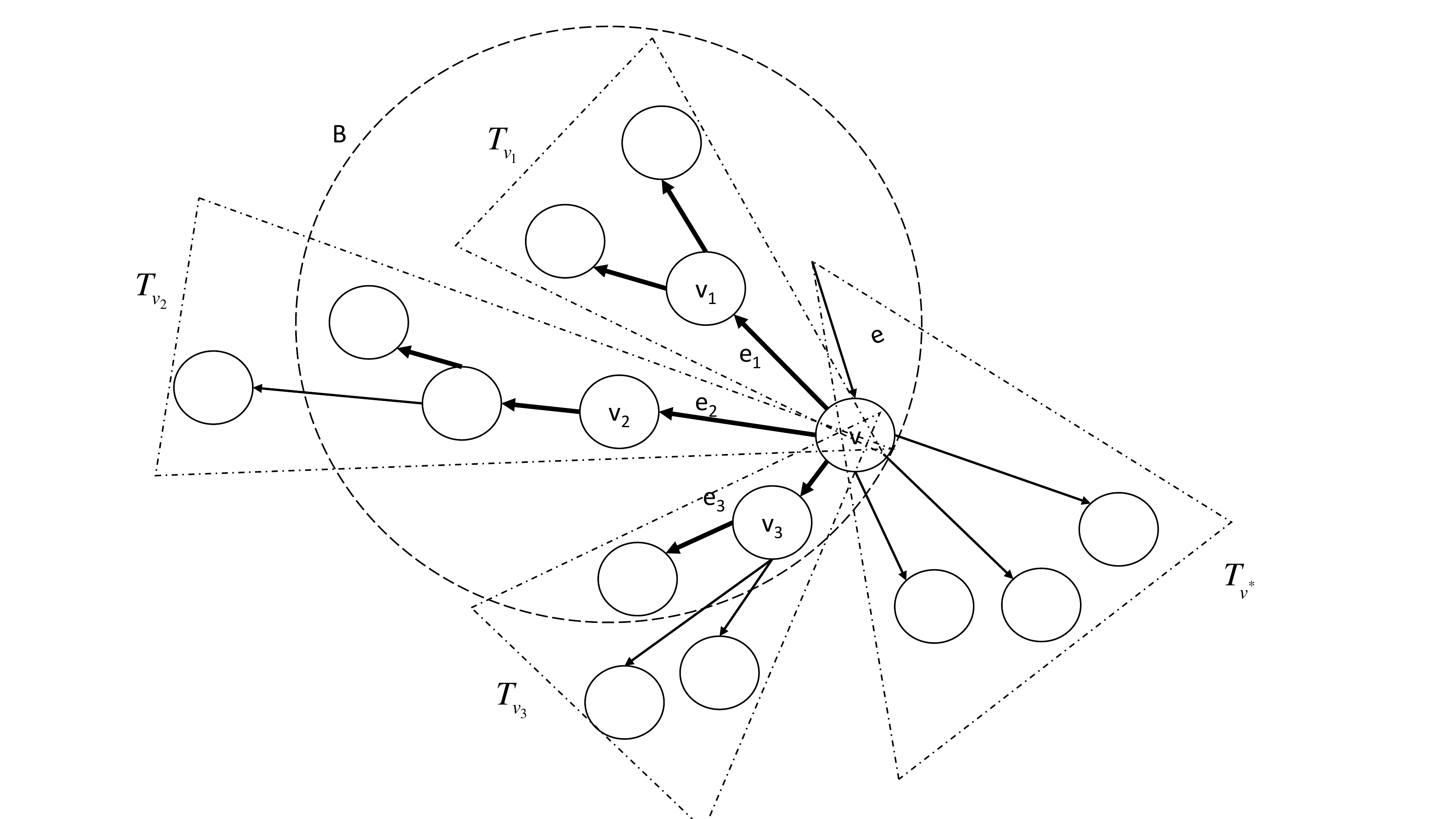}
}
\caption{The computation of the value $OPT_\co(v, \hat{T},x)$. The bold arcs depict the arcs of $\hat{T}$.}
\label{fig:AVertexWithinABlock}
\end{center}
\end{figure}

\begin{figure} [htbp]
\begin{center}
\commentfig{
\includegraphics [scale=0.4]{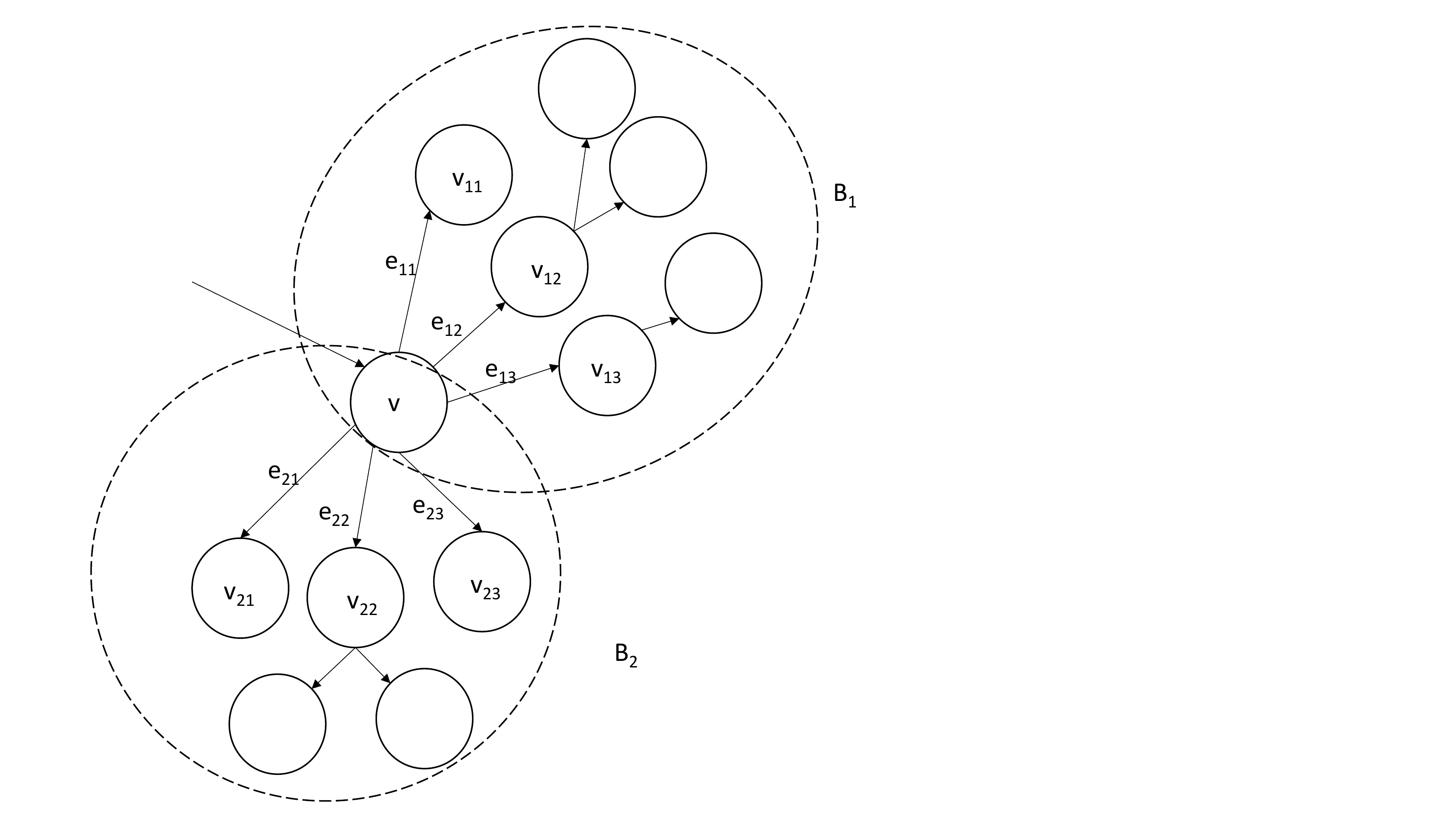}
}
\caption{The computation of the value $OPT_\co(v, x)$ for a cut vertex $v$.}
\label{fig:DynProgACutVertex}
\end{center}
\end{figure}
\end{proof}

\alglanguage{pseudocode}
\begin{algorithm}[H]
\caption{Dynamic Programming for $\mincc$ in Graphs with Sparse Blocks and Bounded Degree Cut Vertices}\label{alg:SparseBlocks}
\begin{algorithmic}
\State $V(G) \gets V(G) + r'$.
\State $E(G) \gets E(G) + rr'$.
\State $B_r \gets \set{r,r'}$.
\State $X \gets X + 0$.
\ForAll {$x \in X$}
\State $\tc(0,x) \gets 0$.
\EndFor
\State $\TT \gets $ the block tree of $G$ rooted at $B_r$
\Statex
\ForAll {vertices $U \in V(\TT) \setminus \set{B_r}$ in a postorder traversal of $\TT$}
    \If {$U$ is a block $B$ of $G$}
        \ForAll {$\hat{T} \in \cs(B)$ rooted at $v(B)$}
            \State \Call{ComputeBlock}{$B, \hat{T}$}
        \EndFor
    \Else \Comment $U$ is a cut vertex $v$ of $G$
        \ForAll {$x \in X$}
            \State Compute $OPT_\co(v,x)$ using equation (\ref{eqn:SparseBlockDynProgCutVertex})
        \EndFor
    \EndIf
\EndFor
\State \Return $OPT_\co(r,0)$.
\Statex
\Statex
\Function {ComputeBlock}{$B, \hat{T}$}
    \ForAll {vertices $v \in V(\hat{T}) \setminus \set{v(B)}$ in a postorder traversal of $\hat{T}$}
        \ForAll {$x \in X$}
            \State Construct a complete bipartite graph $H$ with
            \State ~~~~~~~~~~bipartition $\set{\children{\hat{T}}{v},X-x}$
            \For {$v' \in \children{\hat{T}}{v}$}
                \For {$y \in X - x$}
                    \State $w(v'y) \gets \tc(x,y) + OPT_\co(v',\hat{T},y)$
                \EndFor
            \EndFor
            \State $OPT_\co(v, \hat{T},x) \gets \gamma(H,w)$ \Comment The min. weight perfect matching
            \If {$v$ is a cut vertex of $G$}
                  \State $OPT_\co(v, \hat{T},x) \gets OPT_\co(v, \hat{T},x) + OPT_\co(v,x)$.
            \EndIf
        \EndFor
    \EndFor
\EndFunction
\end{algorithmic}
\end{algorithm}

\runningtitle{Algorithm for a cut vertex}
We proceed with the description of the computation of the values $OPT_\co(v,x)$ for a cut vertex $v$ of $G$, with $\children{\TT}{v}=\set{B_1,\ldots,B_k}$ (see Figure \ref{fig:DynProgACutVertex}). We perform an exhaustive search by guessing a) the coloring $\chi_v$ that an optimal solution $\chi$ induces on the edges incident to $v$, and b) the spanning trees $\hat{T}_i = T[B_i]$ for every child block $B_i$ of $v$. Note that $v$ is the root of all the trees $\hat{T}_i$. Let $\children{\hat{T}_i}{v}= \set{v_{i1},\ldots,v_{i{k_i}}}$ and $e_{ij}=vv_{ij}$. Then, recalling that $\chi(\inn{T}{v})=x$, we have
\begin{eqnarray*}
\co_\chi(T_v) & = & \sum_{i=1}^k \sum_{j=1}^{k_i} \left( \tc(x,\chi_v(e_{ij})) + \co_\chi(T_{v_{ij}}) \right)\\
& = & \sum_{i=1}^k \sum_{j=1}^{k_i} \tc(x,\chi_v(e_{ij})) + \sum_{i=1}^k \sum_{j=1}^{k_i} \co_\chi(T_{v_{ij}}).
\end{eqnarray*}
Given a guess for $\chi_v$ and $\hat{T}_i$, the first sum is fixed and the trees $T_{v_{ij}}$ are pairwise disjoint. Therefore, each term in the second summation can be minimized independently. Then, the minimum for a given guess is
\begin{eqnarray*}
& & \sum_{i=1}^k \sum_{j=1}^{k_i} \tc(x,\chi_v(e_{ij})) + \sum_{i=1}^k \sum_{j=1}^{k_i} \min \co_\chi(T_{v_{ij}})\\
& = & \sum_{i=1}^k \sum_{j=1}^{k_i} \tc(x,\chi_v(e_{ij})) + \sum_{i=1}^k \sum_{j=1}^{k_i} OPT_\co(v_{ij},\hat{T}_i,\chi_v(e_{ij})) \\
& = & \sum_{i=1}^k \sum_{j=1}^{k_i} \left( \tc(x,\chi_v(e_{ij})) + OPT_\co(v_{ij},\hat{T}_i,\chi_v(e_{ij})) \right).
\end{eqnarray*}
For a given guess of $\chi_v$, the minimum over all guesses of the subtrees $\hat{T}_i$ is
\begin{eqnarray*}
& & \sum_{B \in \children{\TT}{v}} \min_{\hat{T} \in \cs(B)} \sum_{v' \in \children{\hat{T}}{v}} \left( \tc(x,\chi_v(vv')) + OPT_\co(v',\hat{T},\chi_v(vv')) \right).
\end{eqnarray*}
Minimizing over all guesses of $\chi_v$ we get
\begin{eqnarray}
OPT_\co(v,x)=\min_{\chi_v \in \ff(S_v)} && \sum_{B \in \children{\TT}{v}} \min_{\hat{T} \in \cs(B)} \sum_{v' \in \children{\hat{T}}{v}} \nonumber\\
&& \left( \tc(x,\chi_v(vv')) + OPT_\co(v',\hat{T},\chi_v(vv')) \right),\label{eqn:SparseBlockDynProgCutVertex}
\end{eqnarray}
where $\ff(S_v)$ is the set of all colorings $\chi_v$ of the edges incident to $v$ using colors from $X$ with the exception that $\chi_r(rr')=0 \notin X$. We note that for the $\minrcpt$ problem the innermost term on the right hand side becomes
$\tc(x,\chi_v(vv')) \cdot \abs{\pp_{vv'}} + OPT_\rc(v',\hat{T},\chi_v(vv'))$.

The number of terms in the inner summation is at most the degree $d(v)$ of the cut vertex $v$, which is at most $c_1$. The number of guesses is at most $\abs{\ff(S_v)} \cdot \abs{\cs(B_i)} \leq \abs{X}^{d(v)} \cdot \abs{E(G)}^{c_2+1} \leq \abs{X}^{c_1} \cdot \abs{E(G)}^{c_2+1}$. Therefore, $OPT_\co(v,x)$ can be computed in time polynomial in the input size.
\qed


\small

\end{document}